\documentclass{article}

\usepackage{graphicx,amsmath,psfrag} 
\usepackage{algorithm}
\usepackage{algpseudocodex}

\usepackage[leftcaption]{sidecap}
\usepackage{todonotes}
\usepackage{mathtools}
\usepackage{complexity}
\usepackage[misc]{ifsym}

\usepackage{xcolor}
\usepackage{amssymb}
\usepackage{cleveref}

\DeclareMathOperator{\kRep}{\mathit{k}\mathbf{Rep}}

\DeclareMathOperator{\MkRep}{\mathbf{M}\mathit{k}\mathbf{Rep}}
\newcommand{\numberofocc}[2]{\#_{#1} (#2)}
\newcommand{\intcompositons}[2]{\Delta_{#1, #2}}
\newcommand{\prevpt}[3]{Prev_{#1}(#2, #3)}
\newcommand{\nextpt}[3]{Next_{#1}(#2, #3)}

\newcommand{\Z}{\mathbb{Z}}
\newcommand{\krep}[1]{\kRep(#1)}
\newcommand{\maxkrep}[1]{\MkRep(#1)}
\newcommand{\intv}[2]{[#1,\;#2]}
\newcommand{\rointv}[2]{[#1,\;#2)}
\newcommand{\lointv}[2]{(#1,\;#2]}
\newcommand{\ointv}[2]{(#1,\;#2)}

\newcommand{\sigmastart}{\(\sigma\)-start}
\newcommand{\algoextendkrep}{\textsc{Extend\(k\)RepSubseq}{\((S, X, \sigma)\)}}
\newcommand{\algoextendcs}{\textsc{M\(k\)CS}}
\newcommand{\algoenumsigstrt}{\textsc{EnumPotential-$\sigma$-Start}}

\newcommand{\sigmasplitpoint}{\(\sigma\)-split point}
\newcommand{\eg}{e.g.\@}

\newcommand{\qed}{\hfill\hbox{\rlap{$\sqcap$}$\sqcup$}}
\newenvironment{proof}{\noindent \emph{Proof.\,}}{\qed}

\newtheorem{example}{Example}
\newtheorem{theorem}{Theorem}[section]
\newtheorem{lemma}{Lemma}[section]
\newtheorem{corollary}{Corollary}[section]
\newtheorem{definition}{Definition}[section]

\newtheorem{observation}{Observation}[section]

\title{Computing Maximal Repeating Subsequences in a String}

\author{Mingyang Gong \footnote{Gianforte School of Computing, Montana State University, Bozeman, MT 59717, USA. Email: {\tt mingyang.gong@montana.edu}.}
\and
Adiesha Liyanage \footnote{Gianforte School of Computing, Montana State University, Bozeman, MT 59717, USA. Email: {\tt a.liyanaralalage@montana.edu}.}
\and
Braeden Sopp \footnote{Gianforte School of Computing, Montana State University, Bozeman, MT 59717, USA. Email: {\tt braeden.sopp@gmail.com}.}
\and
Binhai Zhu \footnote{Gianforte School of Computing, Montana State University, Bozeman, MT 59717, USA. Email: {\tt bhz@montana.edu}.}
}

\date{}

\begin{document}
\maketitle

\begin{abstract}
In this paper we initiate the study of computing a maximal
(not necessarily maximum) repeating pattern in a single input string,
where the corresponding problems have been studied (e.g., a
maximal common subsequence) only in two or more input strings by Hirota and Sakai starting 2019. Given an input string $S$ of length $n$, we can compute
a maximal square subsequence of $S$ in $O(n\log n)$
time, greatly improving the $O(n^2)$ bound for computing the
longest square subsequence of $S$. For a maximal $k$-repeating subsequence, our bound is $O(f(k)n\log n)$, where \(f(k)\) is a computable function such that $f(k) < k\cdot 4^k$. This greatly improves the $O(n^{2k-1})$ bound for computing a longest $k$-repeating subsequence of $S$, for $k\geq 3$. Both results hold for the constrained case, i.e., when the solution must contain a subsequence $X$ of $S$, though with higher running times.
\end{abstract}

\section{Introduction}

Computing (non-contiguous)  patterns in two or more input strings is a classical problem which has found many applications. The textbook example is the longest common subsequence (LCS) of two
strings of a total length of $n$, which was solved by Wagner and Fischer in $O(n^2)$ time and space \cite{DBLP:journals/jacm/WagnerF74} and then in $O(n^2)$ time but $O(n)$ space by Hirschberg \cite{DBLP:journals/cacm/Hirschberg75}. It has been researched and applied in various applications, sometimes with additional
constraint, e.g., the LCS must also contain a pattern $X$, which is a subsequence appearing in both input strings $A$ and $B$. If $A$, $B$ and $X$ are all of length $O(n)$, Tsai gave an $O(n^5)$ time algorithm \cite{DBLP:journals/ipl/Tsai03}. The bound was improved to $O(n^3)$ by Chin et al. \cite{DBLP:journals/ipl/ChinSFHK04} and also by Chen and Chao \cite{DBLP:journals/jco/ChenC11}, the latter additionally considered the version where the LCS of $A$ and $B$ must not include $X$ as a subsequence. 

A few years ago, based on the observed phenomena in applications that sometimes the longest common subsequence might not capture the truely needed or observed pattern in many applications, Sakai
studied the {\em maximal common subsequence (MCS)} of two strings \cite{Sakai18,Sakai19}. For two strings $A$ and $B$ with a total length of $n$, Sakai can compute an MCS of $A$ and $B$ in
$O(n\log n)$ time \footnote{The claimed time bounds in \cite{Sakai18,Sakai19} are incorrect as the construction cost of the used data structure was not counted \cite{DBLP:journals/tcs/HirotaS25}}. Later, given $m$ strings with a total length of $n$, Hirota and Sakai presented an algorithm which can compute an MCS in $O(mn\log n)$ time \cite{HirotaSakai23}. Both of these algorithms work for the
case when the solution must contain a pattern (or subsequence) of $A$ and $B$.

In many biological applications, sometimes useful non-contiguous patterns exist in a single input string. For example, plants go through up to three whole genome duplications; hence there is a general
repetition pattern (even after some subsequent local mutations) \cite{zheng2009gene}. Applications like this motivate the {\em Longest Square Subsequence (LSS)} problem, i.e., a longest pattern like $AA$ where both $AA$ and $A$ itself are subsequences of an input sequence $S$ of length $n$, for which Kosowski gave an $O(n^2)$ time algorithm in 2004 \cite{Kosowski04}. For the longest $k$-repeating subsequence, i.e., a longest pattern like $A^k$ where both
$A^k$ and $A$ are subsequence of $S$, there is a trivial
algorithm: just enumerate $k-1$ cuts in $S$ and then compute the
longest common subsequence of the resulting $k$ strings. This results in an $O(n^{2k-1})$ time algorithm; unfortunately, this is the best known algorithm for $k\geq 3$. For $k=3$, Wang and Zhu gave a parameterized $O(q^2n^3)$ time algorithm, where $q$ is the minimum number of letters deleted in $S$ to have a solution \cite{DBLP:conf/spire/WangZ23}. However, when $q=O(n)$ the worst-case running time of this algorithm remains to be $O(n^{2}\times n^{3})=O(n^5)$.

Recently, motivated by recovering patterns from tandem duplications, Lai et al. studied the {\em Longest Letter-duplicated Subsequence (LLDS)} and related problems \cite{DBLP:journals/acta/LaiLZZ24}. Here a letter-duplicated subsequence (LDS) is a subsequence of $S$ in the form
$\prod_{i=1}^{k} x_i^{d_i}$, where $x_i$ is a letter in $S$ and $d_i$ is an integer at least 2. They showed that LLDS can be solved in $O(n)$ time. In fact, this problem is closely related to the {\em Maximum Run} problem from computational
genomics where each letter $x_i$ can only appear in a run (i.e., consecutively) at most once, and $d_i\geq 1$ \cite{DBLP:conf/wabi/SchrinnerGWSSK20,DBLP:conf/cpm/AsahiroEG0LMOT23}. Very recently, Lafond et al. studied the more general
problem called {\em Longest Subsequence-duplicated Subsequence (LSDS)} and related problems, where a subsequence-duplicated subsequence (SDS) is a subsequence of $S$ in the form $\prod_{i=1}^{k} X_i^{d_i}$, where $X_i$ is a subsequence in $S$ and $d_i$ is an integer at least 2 \cite{DBLP:journals/iandc/LafondLLZ25}. They showed that LSDS can be solved in $O(n^6)$ time, where the fundamental subroutines are to compute the longest square and cubic sequences of (a substring of) $S$.

In this paper, we follow the footsteps of Hirota and Sakai to initiate the research for computing various maximal repeating subsequences in a single string $S$ of length $n$. 
We summarize our results as follows:
\begin{enumerate}
\item A maximal square subsequence of $S$ can be computed in $O(n\log n)$ time.
\item A maximal \(k\)-repeating subsequence of $S$ can be computed
in $O(f(k)\cdot n\log n)$ time for $k\geq 3$ and $f(k)$ is a computable function on $k$, i.e., when $k=O(1)$ the running time is $O(n\log n)$.

\end{enumerate}

The paper is organized as follows. In Section 2 we give necessary definitions. In Section 3 we present the result
to compute a maximal 
 square subsequence of $S$. In Section 4 we
present the result to compute a maximal \(k\)-repeating subsequence of $S$.
In Section 5 we conclude the paper by giving some open problems.

\section{Preliminaries}

\subsection{Basics}
Let $\Sigma$ be an alphabet and $S$ be a given sequence over the alphabet $\Sigma$.
We let $n = |S|$ denote the length of $S$ and $m = |\Sigma|$ denote the size of $\Sigma$, respectively.
For every $1 \le i \le |S|$, we denote by $S[i]$ the $i$-th character in $S$.
A sequence $X$ is a {\em subsequence} of $S$ if there exist $1 \le i_1 < i_2 < \ldots < i_{|X|} \le |S|$ such that $X = S[i_1]S[i_2]\ldots S[i_{|X|}]$; in this case we also abuse the notation by writing $X\subseteq S$.
We denote the empty string as \(\epsilon\).
Given two positive integers $1 \le i < j \le |S|$, we denote by $S[i,j]$ the substring of $S$ that begins with $S[i]$ and ends at $S[j]$. Note that we define \(S[j,\;i] = \epsilon\) when $j>i$.
Furthermore,
$S[1,j]$ is called a {\em prefix} of $S$, while $S[i,n]$ is called a {\em suffix} of $S$.
Given a character $\sigma \in \Sigma$, we use $\sigma^\ell$ to denote the concatenation of $\ell$ copies of the character $\sigma$.
We use open interval notation to indicate exclusive index substrings; i.e., 
\( S\rointv{i}{j} = S\intv{i}{j-1} \), 
\( S\lointv{i}{j} = S\intv{i+1}{j} \), and 
\( S\ointv{i}{j} = S\intv{i+1}{j-1} \).
Moreover, we define \(\numberofocc{S}{\sigma}\) to be the number of occurrences of \(\sigma\) in \(S\).

Given two sequences $S$ and $T$, a maximal common subsequence
$C$ of $S$ and $T$ is one which is not a subsequence of any other common subsequence $C'$ of $S$ and $T$. (Note that the longest common sequence of $S$ and $T$ is simply the longest maximal common subsequence.) Sakai's algorithm is based on the
following observation (which was not explicitly stated in \cite{Sakai18}).

\begin{observation}
\label{obser00}
Let $C$ be a common subsequence of two input strings $S$ and $T$. Either $C$ is maximal, or $C$ can be augmented into a maximal common subsequence $C'$ of $S$ and $T$ by inserting letters.
\end{observation}

Note that this property does not hold for the longest common subsequences, i.e., one cannot always augment an arbitrary common subsequence into a longest one. Sakai's algorithm \cite{Sakai18} basically greedily searches a common subsequence of $S$ and $T$ from right to left and uses additional maximality property to check if it is maximal; if not, then backtrack by adding more letters until the maximality property holds (i.e., a maximal common subsequence is computed).

\subsection{Basics on One String}

A {\em square subsequence} of $S$ is a subsequence $S$ and can be represented in the form $XX$.
A square subsequence $B$ of $S$ is {\em maximal} if there is no other square subsequence $C$ of $S$ that properly contains $B$.
As a warm-up, we first study the {\em maximal square subsequence} problem that seeks a maximal square subsequence of $S$.
One sees that if a character only appears once, then we can directly remove it from the sequence since it will never appear in a maximal square subsequence.
Therefore, we can assume that every character appears at least twice in $S$.

\begin{definition}
\label{def01}
Suppose that $X = YY$ is a square subsequence of $S$.
We say $X$ is {\em right-extendable} ({\em left-extendable}, respectively) 
if there is a character $y$ such that $YyYy$ ($yYyY$, respectively) is a subsequence of $S$.

Similarly, $X = YY$ is {\em inner-extendable} if there is a character $y$ such that $Y_1yY_2 Y_1yY_2$ is subsequence of $S$ where $Y = Y_1Y_2$ and neither of $Y_1$ and $Y_2$ is $\epsilon$.
\end{definition}

One sees that $X$ is a maximal square subsequence of $S$ if and only if $X$ is not right-, left- and inner-extendable.
In fact, the following observation is crucial for our algorithm for the maximal square subsequence problem (which could be thought as an extension of Observation~\ref{obser00}).

\begin{observation}
\label{obser01}
Let $\sigma$ be a letter in $S$ which appears $\ell$ times. Then there must be a maximal square subsequence containing
$2\lfloor \ell/2\rfloor$ copies of $\sigma$.
\end{observation}

\section{Computing a Maximal Square Subsequence}

We present an $O(n\log n)$ time algorithm to compute a maximal square sequence of $S$ (with length $n$). We first arbitrarily choose a character $\sigma \in \Sigma$ and afterwards, record the positions of each character $\sigma$ in $S$ in $O(n)$ time by scanning the sequence $S$ once.
%
The sequence of positions of $\sigma$ in $S$ is denoted as $\langle i_1, i_2, \ldots, i_\ell \rangle$ where $\ell$ is the number of occurrences of $\sigma$ in $S$.
For the ease of presentation, let $e = \lfloor \frac \ell 2 \rfloor$.

Our algorithm starts with the square subsequence $X_0 = \sigma^{2e}$
and firstly converts it into a square subsequence $X_1$ that contains $X_0$; and, the sequence $X_1$ is not right-extendable and inner-extendable.
Finally, we enumerate all possible ``leftmost anchors" of $X_1$ (see \cref{def02}) and; compute a square subsequence $X_2$ that contains $X_1$ and is not left-extendable.
Consequently, the final square subsequence produced $X_2$ is maximal. Figure~\ref{fig1} gives a flowchart of the algorithm for computing a maximal square subsequence of $S$.

\begin{figure}[htbp]
\begin{equation*}
\quad
X_0 =\sigma^{e}\cdot \sigma^{e} \quad\longrightarrow \quad X_{1} = YY \quad\longrightarrow \quad X_{2}=ZZ
\end{equation*}
\caption{A flowchart of the algorithm.}
\label{fig1}
\end{figure}

We next introduce the steps to form $X_1$ from $X_0$.
For the ease of presentation, the notation $MCS(S_1, S_2, S_3)$ is used to denote the maximal common subsequence of $S_1$, $S_2$ that contains $S_3$,
which is outputted by Sakai's algorithm.
We first compute $Y =MCS(S[i_1, i_{e+1}), S[i_{e+1}, n], \sigma^e)$.
Note that $Y$ must be a subsequence of $S[i_1, i_{e+2}]$ since $i_{e+2} > i_{e+1}$ and $Y$ is a subsequence of $S[i_1, i_{e+1}]$.
If $\ell$ is odd and $Y$ is a subsequence of $S[i_{e+2}, n]$, then we update $Y$ by $MCS(S[i_1, i_{e+2}), S[i_{e+2}, n], Y)$.
The final produced $X_1 = YY$ is the concatenation of two copies of $Y$.
We conclude the process of computing $X_1$ in \cref{Approx1} and use an example to illustrate our algorithm.

\begin{algorithm}
\caption{Algorithm for computing $X_1$}
\label{Approx1}
\begin{algorithmic}[1]
\State Input: A sequence $S$, a character $\sigma$ which repeats $\ell$ times, and $e=\lfloor \ell/2\rfloor$.

\State Output: A square subsequence $X_1$ that contains $\sigma^{2e}$ and is not right- and inner-extendable.

\State $Y =MCS(S[i_1, i_{e+1}), S[i_{e+1}, n], \sigma^e)$.

\If {$\ell$ is odd and $Y$ is a subsequence of $S[i_{e+2}, n]$}

\State $Y =MCS(S[i_1, i_{e+2}), S[i_{e+2}, n], Y)$.

\EndIf

\State Return $X_1 = YY$.
\end{algorithmic}
\end{algorithm}

\begin{example}
\label{ex01}
We assume that the input sequence $S = abcabcaccabcac$ and the character $a$ is picked up by the algorithm.
The character $a$ appears $\ell = 5$ times and the sequence of positions of $a$ is $\langle 1, 4, 7, 10, 13 \rangle$.
Therefore, $e = \lfloor \frac 52 \rfloor = 2$, $i_{e+1} = 7$ and $i_{e+2} = 10$.

Since $\ell$ is odd, then we first compute a $MCS(abcabc, accabcac, aa)$ and a possible output $Y$ can be $abcac$.
Then we run $MCS(abcabcacc, abcac, abcac)$ and the final output is $X_1 = YY = abcac \cdot abcac$.
\end{example}

\begin{lemma}
\label{lemma01}
Given any sequence $Z$, $Y$ is a subsequence of $Z$ if and only if $YY$ is a subsequence of $ZZ$.
\end{lemma}
\begin{proof}
Clearly, if $Y$ is a subsequence of $Z$, then $YY$ is a subsequence of $ZZ$.
On the other hand, suppose that $YY$ is a subsequence of $ZZ$ and then the first $Y$ must be a subsequence of $ZZ$.
We construct a $Y$ from $ZZ$ by greedily concatenating the leftmost characters in $ZZ$.
Therefore, $Y$ is either a subsequence of the first $Z$ or $Y$ is a supersequence of the first $Z$.
In the latter case, the second $Y$ is a subsequence of a suffix of the second $Z$ --- then $Y$ must be a subsequence of $Z$. Hence the lemma is proved.
\qed
\end{proof}

\begin{lemma}
\label{lemma02}
The first character of $Y$ is exactly $\sigma$, i.e., $Y[1] = \sigma$ where $X_1 = YY$.
\end{lemma}
\begin{proof}
We prove the lemma by discussing the parity of $\ell$.
If $\ell$ is even, then $Y$ contains $\sigma^e$ and is the maximal common subsequence of $S[i_1, i_{e+1})$ and $S[i_{e+1}, n]$.
In this case, every $\sigma$ in $S[i_1, i_{e+1}-1]$ must appear in $Y$ and the lemma is proved since the first character of $S[i_1, i_{e+1})$ is exactly $\sigma$.

We next assume $\ell$ is odd.
If $Y$ is a maximal common subsequence of $S[i_1, i_{e+1})$ and $S[i_{e+1}, n]$, then similarly to an even $\ell$, we are done.
Otherwise, $Y$ is a maximal common subsequence of $S[i_1, i_{e+2})$ and $S[i_{e+2}, n]$ and; $Y$ contains $\sigma^e$.
In this case, every $\sigma$ in $S[i_{e+2}, n]$ must appear in $Y$.
The lemma follows since the first character of $S[i_{e+2}, n]$ is exactly $\sigma$.
\qed
\end{proof}

\begin{lemma}
\label{lemma03}
A square subsequence that contains $X_1$ must not be right-extendable.
\end{lemma}
\begin{proof}
We assume that a square subsequence $ZZ$ contains $X_1 = YY$.
Suppose that $ZZ$ is right-extendable and thus there exists a character $y$ 
such that $ZyZy$ is a square subsequence of $S$.
By \cref{lemma01} and $\sigma^{2e} \subseteq YY \subseteq ZZ$, $\sigma^e$ is a subsequence of $Y$ and $Y$ is a subsequence of $Z$.
Therefore, $YyYy$ is a square subsequence of $S$, that is, $YY$ is also right-extendable.
We next consider the following two cases.

Case 1: $\ell$ is even. 
By \cref{lemma02}, $Y[1]=\sigma$.
Since $Y$ contains $\sigma^e$, the character $Y[1]=\sigma$ in the second $Yy$ of $YyYy$ must appear in the position $i_{e+1}$.
In other words, the first $Yy$ of $YyYy$ is a subsequence of $S[i_1, i_{e+1})$ and the second $Yy$ is a subsequence of $S[i_{e+1}, n]$, respectively.
Therefore, $Yy$ is a common subsequence of $S[i_1, i_{e+1})$ and $S[i_{e+1}, n]$,
which contradicts the maximality of $Y$.

Case 2: $\ell$ is odd.
By \cref{lemma02} and $Y$ contains $\sigma^e$, the character $Y[1]=\sigma$ in the second $Yy$ of $YyYy$ must appear in the position $i_{e+1}$ or $i_{e+2}$.
In other words, $Yy$ must be a common subsequence of $S[i_1, i_{e+2})$ and $S[i_{e+2}, n]$ or;
a common subsequence of $S[i_1, i_{e+1})$ and $S[i_{e+1}, n]$.

Case 2.1: $Yy$ is a common subsequence of $S[i_1, i_{e+2})$ and $S[i_{e+2}, n]$ and so is $Y$.
Therefore, $Y$ is computed in line 5 of \cref{Approx1}. 
But because of the existence of $Yy$, $Y$ is not a maximal common subsequence of $S[i_1, i_{e+2})$ and $S[i_{e+2}, n]$, a contradiction.

Case 2.2: $Yy$ is a common subsequence of $S[i_1, i_{e+1})$ and $S[i_{e+1}, n]$.
Let $Y'$ be the produced sequence in line 3 of \cref{Approx1}.
If $Y'$ is not a subsequence of $S[i_{e+2}, n]$, then $Y = Y'$.
Otherwise, $Y$ is a supersequence of $Y'$.
Therefore, in both two cases, $Y'$ is properly contained in $Yy$, which contradicts the maximality of $Y'$.

The lemma is proved.
\qed
\end{proof}

\begin{lemma}
\label{lemma04}
$X_1 = YY$ is not inner-extendable.
\end{lemma}
\begin{proof}
Suppose it is not the case and thus there exists a square subsequence $Y_1yY_2 Y_1yY_2$ such that $Y = Y_1Y_2$ and none of $Y_1, Y_2$ is $\epsilon$.
Since $Y_1$ is not $\epsilon$, by \cref{lemma02}, $Y_1[1] = X_1[1] = \sigma$.
Similar to \cref{lemma03}, the lemma can be proved by discussing the parity of $\ell$.
\qed
\end{proof}

By \cref{lemma02}, $X_1[1]=Y[1]=\sigma$ and each copy $Y$ in $X_1$ contains exactly $e$ $\sigma$'s.
Therefore, the first $\sigma$ in $X_1$ must appear in the position $i_1$ when $\ell$ is even 
and in the positions $i_1, i_2$ when $\ell$ is odd.

\begin{definition}
[Leftmost anchor with a position $i_t$]
\label{def02}
Suppose that $Y$ is the sequence produced by \cref{Approx1}.
Given a position $i_t$ of $\sigma$, a {\em leftmost anchor of $Y$ with $i_t$} is the smallest index $j_t$ such that $Y$ is a subsequence of $S[i_t, j_t]$.
\end{definition}

Note that the index $j_t$ can be determined in $O(n)$ by finding the leftmost $Y$ starting with the position $i_t$ for $t\in\{1,2\}$. It might not exist when $t\geq 3$.

\begin{example}
\label{ex02}
Consider the input sequence $S = abcabcaccabcac$ and the character $a$ in \cref{ex01}.
Recall that $Y = abcac$.
The leftmost anchor of $Y$ with the position $i_1 = 1$ and $i_2 = 3$ is $j_1 = 6$ and $j_2 = 8$, respectively.
\end{example}

The following algorithm takes $X_1 = YY$ produced by \cref{Approx1} as its input and 
outputs a maximal square subsequence $X_2 = ZZ$, which is a supersequence of $X_1$.
We first determine the leftmost anchor of $Y$ with the position $i_1$ by \cref{def02} 
where $i_1$ is the first position of $\sigma$ in $S$.
Note that $j_1$ must exist since $Y$ stars with a $\sigma$ and $Y$ is a subsequence of $S$.
Then we compute $Z =MCS(S[1, j_1], S(j_1, n], Y)$.
We remind the readers that in this case, $S[1, j_1]$ starts with the first position of $S$ but in \cref{Approx1}, we always starts with the first position of $\sigma$ in $S$.
The reason is that we hope to contain as many characters that is to the left side of $\sigma$ as possible
to guarantee the final produced sequence is not left-extendable.
Therefore, the current $Z$ may not start with a $\sigma$.
Note that $ZZ$ is a subsequence that contains $YY$.
By Lemmas~\ref{lemma01}, \ref{lemma03} and~\ref{lemma04}, $Z$ is in fact a concatenation of some sequence $Z_1$ and $Y$, i.e., $Z = Z_1 Y$.

We next show the leftmost anchor of $Z$ with the position $i_2$ must exist.
Recall that $Z$ is a subsequence of $S(j_1, n]$ and so is $Y$.
Since $Y$ starts with a $\sigma$ and $j_1 \ge i_1$, $Y$ must be a subsequence of $S[i_2, n]$ and thus $j_2$ exists.

One sees that by \cref{def02}, $j_2 \ge j_1$ and thus $Z$ must be a subsequence of $S[1, j_2]$.
If $\ell$ is an odd number and $Z$ is a subsequence of $S(j_2, n]$, then we update $Z$ by $MCS(S[1, j_2], S(j_2, n], Z)$.
The final produced $X_2 = ZZ$ is the concatenation of two $Z$ copies.
We conclude with the pseudocode of computing $X_2$ in \cref{Approx2} and in \cref{ex03}, we illustrate the algorithm by using the same instance in \cref{ex01} and \cref{ex02}.

\begin{algorithm}
\caption{Algorithm for computing a maximal square subsequence $X_2$}
\label{Approx2}
\begin{algorithmic}[1]
\State Input: The produced sequence $X_1 = YY$ by \cref{Approx1}.

\State Output: A maximal square subsequence $X_2$ that contains $X_1$.

\State Determine the leftmost anchor $j_1$ of $Y$ with the position $i_1$.

\State $Z =MCS(S[1, j_1], S(j_1, n], Y)$.

\If {$\ell$ is odd}

\State Determine the leftmost anchor $j_2$ of $Y$ with the position $i_2$.

\If {$Z$ is a subsequence of $S(j_2, n]$}

\State $Z =MCS(S[1, j_2], S(j_2, n], Z)$.

\EndIf
\EndIf

\State Return $X_2 = ZZ$.
\end{algorithmic}
\end{algorithm}

\begin{example}
\label{ex03}
Consider the input sequence $S = abcabcaccabcac$ in \cref{ex01}.
The output of \cref{Approx1} is $X_1 = YY = abcac \cdot abcac$ where $Y = abcac$ and the character $a$ is picked up by the algorithm.
By \cref{ex02}, $i_1 = 1$ and $j_1 = 6$.
Therefore, the sequence $Z$ in Line 4 of \cref{Approx2} is outputted by $MCS[abcabc, accabcac, abcac]$,
which remains $abcac$.

Note that $j_2 = 8$ and $Z = abcac$ is a subsequence of $S(j_2, n] = cabcac$.
Therefore, the final $Z$ is computed by $MCS[abcabcac, cabcac, abcac]$ and $Z = cabcac$.
Thus, Algorithm \cref{Approx2} gives us a maximal square subsequence $X_2 = ZZ = cabcac \cdot cabcac$.
In this example, we can see that Line 8 is necessary to guarantee $X_2$ is maximal.
\end{example}

\begin{lemma}
\label{lemma05}
The outputted sequence $X_2 = ZZ$ is not left-extendable.
\end{lemma}

\begin{proof}
Recall that $ZZ$ must be a square subsequence that contains $YY$.
By Lemmas~\ref{lemma01}, \ref{lemma03} and~\ref{lemma04}, $Z$ can be represented in the form $Z_1 Y$ for some sequence $Z_1$.

We assume the lemma does not hold and thus, there exists a character $y$ such that $yZ_1Y yZ_1Y$ is a square subsequence of $S$ where $Z = Z_1Y$.

Case 1: $\ell$ is even.
Note that $Y$ is a subsequence of $S[i_1, j_1]$ and the first copy $yZ_1 Y$ of $yZ_1Y yZ_1Y$ must be a subsequence of $S[1, j_1]$.
It follows that $yZ = yZ_1Y$ is a common subsequence of $S[1, j_1]$ and $S(j_1, n]$, which contradicts the maximality of $Z$.

Case 2: $\ell$ is odd.
In this case, since $Y$ contains $\sigma^e$, the first $\sigma$ of $Y$ in the first copy $yZ_1Y$ in $yZ_1Y yZ_1Y$ must appear in the position $i_1$ or $i_2$.
Therefore, the first copy $yZ_1Y$ must be a common subsequence of $S[1, j_1]$ and $S(j_1, n]$ or;
$S[1, j_2]$ and $S(j_2, n]$.

Case 2.1: $yZ_1Y$ is a common subsequence of $S[1, j_2]$ and $S(j_2, n]$.
Therefore $Z = Z_1Y$ is a subsequence of $S(j_2, n]$.
It follows that $Z$ must be computed in Line 8 of \cref{Approx2}, which is impossible since $Z$ is not a maximal common subsequence of $S[1, j_2]$ and $S(j_2, n]$.

Case 2.2: $yZ_1Y$ is a common subsequence of $S[1, j_1]$ and $S(j_1, n]$.
For ease of presentation, we let $Z'$ be the sequence computed by Line 4 of \cref{Approx2}.
If $Z'$ is not a subsequence of $S(j_2, n]$, then $Z = Z'$.
Otherwise, $Z'$ is a subsequence of $Z$.
In conclusion, $Z'$ is a subsequence of $Z$, which is properly contained in $yZ = yZ_1Y$ and contradicts the maximality of $Z'$.

The lemma is proved.
\qed
\end{proof}

\begin{lemma}
\label{lemma06}
The outputted sequence $X_2 = ZZ$ is not inner-extendable.
\end{lemma}
\begin{proof}
Suppose not and thus there exists a character $y$ such that $Z_1yZ_2 Z_1yZ_2$ is a square subsequence where $Z = Z_1 Z_2$.
By Lemmas~\ref{lemma03} and~\ref{lemma04}, $Z_2$ must contain $Y$ and in fact, equal to $Z_3 Y$, for some $Z_3$.
In other words, $Z_1yZ_3Y Z_1yZ_3Y$ is a square subsequence and $Z = Z_1Z_3Y$.
The remaining proof is similar to that of \cref{lemma05}.
\qed
\end{proof}

\begin{theorem}
\label{thm01}
Given a sequence $S$ of length $n$, the outputted sequence $X_2 = ZZ$ is a maximal square subsequence of $S$ which can be
computed in $O(n\log n)$ time.
\end{theorem}
\begin{proof}
The theorem is proved by $YY \subseteq ZZ$ and Lemmas~\ref{lemma03},~\ref{lemma05} and~\ref{lemma06}.
The running time of the computation is due to that we call
Sakai's $O(n\log n)$ time algorithm \cite{Sakai18} at most four times.
\qed
\end{proof}

In the next section, we consider the maximal $k$-repeating subsequence problem.

\section{Computing a Maximal $k$-repeating Subsequence}
\subsection{Definitions}

\begin{definition}[alignment]
Given two strings \(A\) and \(S\), an alignment of \(A\) on \(S\) is a one-to-one function \(\phi\) from \(\intv{1}{|A|}\) to \(\intv{1}{|B|}\) such that the order of the indices in \(A\) are preserved by \(\phi\) and the letter at the every index of \(A\) is the same as letter in \(S\) its index gets mapped to.
Note that \(A\) is subsequence of \(S\) if and only if there is some alignment of \(A\) on \(S\).
\end{definition}

\begin{definition}[\(k\)-repeating subsequence]
Let \( S \) be a string and let \( k \in \Z_{>0} \). 
A subsequence \( X \subseteq S \) is called a \emph{\(k\)-repeating subsequence} of \( S \) 
if and only if \( X^k \subseteq S \). 
We denote the set of all \( k \)-repeating subsequences of \( S \) by \( \krep{S} \).
\label{def:krepsubsequence}
\end{definition}

\begin{definition}[maximal-\(k\)-repeating subsequence]
Let \( S \) be a string and let \( X \in \krep{S} \). 
We say that \( X \) is a \emph{maximal \(k\)-repeating subsequence} of \( S \) 
if and only if for every subsequence \( Y \) with \( Y \supset X \), we have \( Y \notin \krep{S}\). We denote the set of all Maximal-\(k\)-Repeating subsequences of \(S\) by \(\maxkrep{S}\).
\label{def:maximalkrepsubsequence}
\end{definition}


For a string \(S\), an index \(i \in \intv{1}{|S|}\), and \(X \subseteq S\), we use \(\nextpt{S}{X}{i}\) to denote the smallest index \(j \geq i\) such that \(S\lointv{i}{j}\) contains \(X\) as a subsequence. 
If \(X \not\subseteq S\lointv{i}{|S|}\), we define \(\nextpt{S}{X}{i}\) to be \textsc{Null}.
Likewise we use \(\prevpt{S}{X}{i}\) to denote the largest index \(l \leq i\) such that \(S\rointv{l}{i}\) contains \(X\) as a subsequence.
If \(X \not\subseteq S\rointv{1}{i}\), we define \(\prevpt{S}{X}{i}\) to be \textsc{Null}.
Intuitively, \(\prevpt{S}{X}{i}\) and \(\nextpt{S}{X}{i}\) adapt~\cref{def02} to both left and right and to work for any index.

\begin{definition}[\sigmastart]
\label{def:sigmastart} Given the string \(S\) and \(r\in \Z_{>0}\), a {\em $\sigma$-start} for \(\sigma^r\) is a $k$-tuple $(p_1, p_2, \ldots, p_k)$ such that the following three conditions are satisfied:

\begin{itemize}
\item[1.] $1 \le p_1 < p_2 < \ldots < p_k \le |S|$ and;

\item[2.] $S[p_1] = S[p_2] = \ldots = S[p_k] = \sigma$ and;

\item[3.] For the ease of presentation, we let $p_{k+1} = |S|+1$. 
Then $\sigma^{r}$ is a subsequence of $S[p_j, p_{j+1}-1]$ for every $j = 1, 2, \ldots, k$.
\end{itemize}
\end{definition}

\begin{definition}[\sigmasplitpoint]
    Given string \(S\), let \(A, B \in \Sigma^{*}\), the \sigmasplitpoint s~for \(A\) and \(B\) are \(p \in \intv{1}{|S|}^k\) where each
    index in \(p\) corresponds a \(\sigma\) in \(S\) and for every \(i \in \intv{1}{k}\), \(\nextpt{S}{B}{p_i}\) and
    \(\prevpt{S}{A}{p_i}\) are not \textsc{Null}
    with \(\nextpt{S}{B}{p_i} < \prevpt{S}{A}{p_{i+1}}\) when \(i \neq k\). 
    \label{def:sigmasplitpoint}
\end{definition}

Note that a \sigmastart~for \(\sigma^r\) is also a \sigmasplitpoint~for \(\epsilon\) and \(\sigma^{r-1}\); conversely, any such \sigmasplitpoint~corresponds to a \sigmastart~for \(\sigma^r\).
We use these definitions along with the following observation
as one of the key basis of our algorithm.
\begin{observation}
     Let \(X \in \Sigma^{*} \) and \(X = A\sigma B\)
     for \(A, B \in \Sigma^{*}\).
     For any \(X' \in \Sigma^{*}\) such that \(X \subseteq X'\) we have \(X' = A'\sigma B'\)  for some \(A \subseteq A'\) and \(B \subseteq B'\).
     \label{obs:prefix-suffix-extend}
\end{observation}

The above observation informs us that if we cannot
insert characters into the prefix \(A\) or the suffix \(B\) to extend \(X\),
then it is impossible to extend \(X\) itself within the alignment.
In other words, any extension of \(X\) must come from 
extending either \(A\) on the left or \(B\) on the right,
while keeping the central \(\sigma\) fixed.

\subsection{Overview}
In this section we provide a overview of our algorithm. However before we do so
we look at a simpler problem to provide intuition for find \(k\)-repeating subsequences.

\subsubsection{Finding one maximal subsequences in a collection}
Similarly to the Maximal Square Subsequence problem,
when searching for \(k\)-repeat subsequences, we use
the subsequence relation as means to
compare and order strings instead of using their length.
To provide intuition for our algorithm, first consider the
problem of finding a single string in a collection \(S=\{A_1, A_2,\ldots, A_k\}\) with no proper supersequences in \(S\) (\eg, a maximal string in \(S\) with respect to the subsequence relation) as opposed to the longest string.
Suppose we start by selecting some string \(A_i\) in \(S\) which we
treat as candidate for being a maximal string in \(S\).
When we attempt to verify if \(A_i\) is maximal, we either find some element in \(S\) that is a proper supersequence of \(A_i\) (thus \(A_i\) is not maximal) or
we cannot and therefore \(A_i\) is maximal.
In the first case, we are able to insert some characters into \(A_i\) to produce a larger (with respect to subsequence relation) string in \(S\) which we can use as a new candidate for a maximal string in \(S\).
Thus, repeating this process allows us to find some additional characters which we can insert into \(A_i\) to produce another string in \(S\), however, as \(S\) only has a limited number of strings, we eventually
must find a supersequence of \(A_i\) in \(S\) that cannot be improved any further by inserting additional characters.
This means a maximal subsequence of \(S\) can always be found by adding characters to an initial guess.
If we, instead, were searching for the longest string in \(S\),
this may not be possible, as every longest string in \(S\) may
require us to delete parts of \(A_i\) to produce it.

For a concrete example, take \(S = \{aba, ababb, aa, c, xyxyay\}\).
If we start with \(aa\), we can produce a maximal string in \(S\) by first jumping to \(aba\) by inserting \(b\) into \(aa\) and then jumping to \(ababb\) by inserting \(bb\) on the end of \(aba\). 
At this point, see that we can no longer improve \(ababb\) and therefore we have found a maximal string.
However, if we are searching for the longest string in \(S\), it is impossible to obtain \(xyxyay\) starting from \(aa\) merely by inserting additional letters; doing so would require deleting one \(a\) from \(aa\).

The example above highlights another quirk related to searching for maximal strings.
If we instead pick \(c\) as the first candidate, we would have immediately found a maximal string in \(S\) despite \(c\) also being the shortest string in \(S\).
Further, unlike \(aa\), where we may need to through two rounds of searching for letters to insert before reaching a maximal solution, here we would have to go through only one.
Taken together, this means (1) when looking for maximal strings in a collection, we can start by selecting any string and look for characters to insert into it until it becomes a maximal string.
(2) The initial candidate can affect how long it takes us to produce a solution. 
Picking a good first candidate could be dramatically faster even if it is small. 

Before presenting the proposed algorithm, we first describe a trivial approach for finding a \(k\)-repeating subsequence of a string \(S\) of length \(n\). 
Suppose we are given a candidate subsequence \(X \subseteq S\) such that \(X\) is a \(k\)-repeating, i.e., \(X^k \subseteq S\). We can always pick any letter \(\sigma \in \Sigma\) that appears at least \(k\) times in \(S\) as the initial candidate. 
Observe that the length of any \(k\)-repeating subsequence is at most \(n/k\). 
For each position, there are at most \(O(n/k)\) choices for the letter to be inserted.  

Each extended candidate can be verified to be a valid \(k\)-repeating subsequence in \(O(n)\) time.
At every position, we iteratively try each choice, inserting the character after the position until no more character can be inserted then we continue to next position.
Note that adding character in later position could never result in new candidates for position already checked.
Since the candidate subsequence can be extended at most \(O(n/k)\) times, the total running time of the trivial algorithm is \(O\big(\frac{n^3}{k^2}\big)\). 

\subsubsection{Algorithm for \(k\)-Repeating Subsequences}
We provide an overview of how our algorithm works to finds a new candidate for a maximal \(k\)-repeating subsequence given a initial candidate \(X\) in \(S\).
At a high level, our algorithm starts by
picking a ``guess separator" letter \(\sigma\) in \(X\) and writing  \(X\) as \(A\sigma B\)
where \(A\) is the prefix in \(X\) that ends just before the first \(\sigma\) and \(B\) is the suffix in \(X\) that ends just after the first \(\sigma\) in \(S\).
We then attempt to iteratively build up a new prefix \(A'\)
that contains \(A\) and suffix \(B'\) that contains
\(B\) by exhaustively scanning all locations for \(\sigma\) in \(S\) where a supersequence \(A''\sigma B''\) with a equal or larger (with respect to subsequence relation) suffix and prefix could be.
Note if we verify we can not extend our suffix or prefix given any location,~\cref{obs:prefix-suffix-extend} implies we have maximal solution. 
At each location in our scan, we apply the algorithm by Hirota and Sakai \cite{HirotaSakai23} to improve our current \(A'\) and \(B'\) and eliminate the possibility that any \(k\)-repeating subsequence \(A''\sigma B''\) can be found at the current locations such that we could increase the size of \(A'\sigma B'\) by inserting characters into our prefix or suffix and taking the produced string as our new candidate.
After our algorithm has exhausted all locations, we know no characters can be inserted before our prefix and suffix to produce a proper supersequence that is also \(k\)-repeating in \(S\); thus, our current string \(A'\sigma B'\) is maximal.

Now fix a string \( X \in \krep{S} \) which is our initial candidate for the algorithm. We can write \( X = A \sigma B \), where 
\( A \in (\Sigma \setminus \{\sigma\})^* \), \( B \in \Sigma^* \), for any \( \sigma \in \Sigma \) that appears in \( X \).
Let \(r=\numberofocc{X}{\sigma}\) be the number of occurrences of \(\sigma\) in \(X\) and \(R=\numberofocc{S}{\sigma} - k\cdot r\).
Since \( X \in \krep{S} \), there must exist an alignment of \( X^k \) as it is a subsequence of \( S \). 
In particular, the first occurrence of \(\sigma\) in each copy of \( X \) must be mapped to some occurrence 
of \(\sigma\) in \( S \) thus, \(R\) is the number of \(\sigma\) in \(S\) not used in an alignment of \(X^k\).
Intuitively, our algorithm proceeds by enumerating and iterating over all potential alignments of these first \(\sigma\) for our initial candidate.
To do this we consider locations for the \(R\) remaining \(\sigma\)
across the \( k \) blocks of \(r\) \(\sigma\) placed consecutively where the start position of the next block is determined by the position of the last block
placed and some number of the remaining \(R\) \(\sigma\) we choose to not use in a block.
Note that, for any potential alignment of the first \(\sigma\), there must exist at least \(r\) number of \(\sigma\) between the starting positions of the first \(\sigma\) in each consecutive block; otherwise, it cannot be a potential starting alignment of the first \(\sigma\).
However, note that, for the potential positions of first \(\sigma\) we consider, there is no guarantee that we can align \(A\) to the left and \(B\) to the right
of each \(\sigma\) at the position we guess such that their placement forms an alignment for \((A\sigma B)^k\).
This is because the only guarantee is that there are at least \(r\) occurrences of \(\sigma\) in \(S\) between the consecutive starting positions of the first \(\sigma\) in \(k\) blocks of \(r\) \(\sigma\), which is a necessary but not sufficient condition for any first \(\sigma\) in each block of a sequence containing \(\sigma^r\) such as \(X\).
Thus we consider a superset of all positions for the first \(\sigma\) in each block for \(X^k\).
Given a potential alignment of the first \(\sigma\), the algorithm attempts to extend the left and the right 
parts of each block (relative to the first \(\sigma\)) in order to form a larger candidate \( k \)-repeating subsequence.
As a result, the main computational work comes from iterating over all potential alignments of the first \(\sigma\) in \( X \). 
Fortunately, we can bound the number of such alignments using a stars-and-bars argument.

Suppose \(Y\in \maxkrep{S}\) such that \( Y \supseteq X \). 
Then the alignment of some \(\sigma\), although maybe not the first, in each block of \( Y \) must also correspond to some starting alignment 
of the first \(\sigma\) in \( X \). 
This means that by systematically considering all possible alignments of the first \(\sigma\) in \( X \) 
and attempting to extend our current candidate with the algorithm by Hirota and Sakai \cite{HirotaSakai23}, our procedure is guaranteed 
to find a maximal \( k \)-repeating subsequence that contains \( X \). 
To calculate the number of such alignments, observe that if \( X \in \krep{S} \) and contains \(r\)
occurrences of \(\sigma\), then aligning \( k \) copies of \( X \) will consume \( k \cdot r \)
occurrences of \(\sigma\) in \( S \) and the remaining number of unused \(\sigma\) is \(R\).
These remaining \( R \) occurrences can be distributed in the gaps between the locations of the first \(\sigma\) as described earlier, 
which is equivalent to partitioning \( R \) indistinguishable items into \( k+1 \) distinguishable bins.
Clearly, there is a bijection between the set of potential alignments of the first \(\sigma\) in each block 
and the set of such partitions.
This gives us a concrete way to bound the number of guesses the algorithm needs to make 
and, consequently, to derive its worst-case running time.

\subsection{Algorithm and Correctness}
We now produce a series of lemmas that
we will need to show our algorithm is correct as well as its runtime.

We start by showing the number of \(\sigma\)-start we need to produce has
an upper bound. In~\cref{app:enumeratingsigstart}, we
provide a detailed algorithm for quickly enumerating \(\sigma\)-starts.

\begin{lemma}
\label{lemma:EnumeratePotentialcorrectness}
Let \(\ell = \numberofocc{S}{\sigma}\). The number of distinct $\sigma$-starts
\(\sigma^r\) is bounded by $\binom{R+k}{k}$ where \(R = \ell-rk\).
\end{lemma}
\begin{proof}
Note $\ell = kr + R$;
therefore, $s \in \intv{0}{k-1}$.
Let $x_j$ be the number of $\sigma$'s in the $S\intv{p_j}{p_{j+1}-1})$ for every $j = 1, \ldots, k$.
By \cref{def:sigmastart}, we have $\sum_{j=1}^k x_j = \ell$ and $x_j \ge e$ for every $j$.
In other words, if we set up $x'_j = x_j - e$, then $\sum_{j=1}^k x'_j = \ell - ke = s$ and $x'_j \ge 0$.
It is not hard to see that the number of distinct $\sigma$-starts is at most $\binom{R+k}k$.
\qed
\end{proof}

We now provide lemmas relating to \(\sigma\)-split points that follow immediately from their definitions.
First, we relate \(\sigma\)-split points to \(k\)-repeating subsequences.

\begin{lemma}
    Let \(A,B\in \Sigma^{*}\) and \(\sigma \in \Sigma\).
    Then \(A\sigma B\) is a \(k\)-repeating subsequence if and only if there is exist some \(\sigma\)-split point for \(A\) and \(B\).
    \label{lemma:bicond-krep}
\end{lemma}
\begin{proof}
    We start with the reverse direction.
    Let \(p\) be some \(\sigma\)-split point for \(A\) and \(B\).
    As  for any \(i \in \intv{1}{k}\) 
    \(\prevpt{S}{A}{p_i}\)
    and
    \(\prevpt{S}{A}{p_i}\)
 are not \textsc{Null} and
    \(S[i]=\sigma\),
    we have \(A\sigma B \subseteq S\intv{\prevpt{S}{A}{p_i}}{\nextpt{S}{B}{p_i}}\).
    Let \(T_1 = S\intv{\prevpt{S}{A}{p_1}}{\nextpt{S}{A}{p_1}}\) and \(T_{i} = T_{i-1} \cdot S\intv{\prevpt{S}{A}{p_i}}{\nextpt{S}{A}{p_i}}\) for \(i \in \intv{2}{k}\).
    As \(\nextpt{S}{B}{p_i} < \prevpt{S}{A}{p_{i+1}}\) when \(i \neq k\), then 
    we have \(T_i \subseteq S\) for any \(i\).
    Thus \((A\sigma B)^k \subseteq T_k \subseteq S\) and \(A \sigma B\) is a  \(k\)-repeating subsequence.
    
    Now we consider the forward direction. As \(A \sigma B\) is a \(k\)-repeating subsequence of \(S\) we have \((A\sigma B)^k \subseteq  S.\)
    thus there some alignment \(\phi\) from \((A\sigma B)^k\) on \(S\).
    Let \(X = A\sigma B\) then \(p = (\phi(|A\sigma|), \phi(|XA\sigma|),\ldots ,\phi(|X^{k-1}A\sigma|))\). By definition of alignment \(S[\phi(|A\sigma|)] = \sigma\) and \(S[\phi(|X^{i}A\sigma|)] = \sigma\) for all \(i \in \intv{2}{k-1}\).
    Observe that \(\phi(1) \leq\prevpt{S}{A}{p_1}\)
    and \(\nextpt{S}{B}{p_k} \leq \phi(|X^k|)\) as well as 
    \(\phi(1+|X^i|) \leq \prevpt{S}{A}{p_{i+1}}\)
    and
    \(\nextpt{S}{B}{p_{i}}\leq \phi(|X^i|) \)
    for \(i \in \intv{1}{k-1}\).
    Thus our prev and next indices are never \textsc{Null}.
    Finally we have \(\nextpt{S}{B}{p_{i}} < \prevpt{S}{A}{p_{i+1}}\) for any \(i \in \intv{1}{k-1}\) as \(\nextpt{S}{B}{p_{i}} \leq \phi(|X^i|) < \phi(1+|X^i|) \leq \prevpt{S}{A}{p_{i+1}}\).
    Thus \(p\) is a \(\sigma\)-split point for \(A\) and \(B\).
\qed
\end{proof}

    The following lemma relates \(\sigma\)-split points between nested \(k\)-related subsequences.

\begin{lemma}
    Let \(A,A' \in \Sigma^{*}\) and \(B,B' \in \Sigma^{*}\) with \(\sigma \in \Sigma\) such that \(A \subseteq A'\) and \(B \subseteq B'\).
    Any \(\sigma\)-split point \(p\) between \(A'\) and \(B'\) is a \(\sigma\)-split point between \(A\) and \(B\).
    \label{lemma:split-implies-split}
\end{lemma}

\begin{proof}
    Note that when \(A_i = A\) and \(B_i=B\) this lemma holds true trivially.
    Suppose \(p\) is \sigmasplitpoint~for \(A'\) and \(B'\). 
    By \cref{def:sigmasplitpoint}, we know that 
    \(\prevpt{S}{A'}{p_i}\) and \(\nextpt{S}{B'}{p_i}\) are both not \textsc{Null} for all \( i \in \intv{1}{k} \),
    and moreover, \(\nextpt{S}{B'}{p_i} < \prevpt{S}{A'}{p_{i+1}}\) whenever \( i \neq k \).
    Essentially, for every component \(i\in \intv{1}{k}\) in \(p\), there exists some indices \( i_{prev} \leq p_i \leq i_{next}\) such that \(S\rointv{i_{prev}}{p_i} \supseteq A'\) and \(S\lointv{p_i}{i_{next}} \supseteq B'\). 
    Since \(A' \supseteq A\) and \(B' \supseteq B\), it follows that \(S\rointv{i_{prev}}{p_i} \supseteq A\) and \(S\lointv{p_i}{i_{next}} \supseteq B\). 
    As a consequence for each \(i \in \intv{1}{k}\) there exists \(i'_{prev}\) and \(i'_{next}\) such that \(i_{prev} \leq i'_{prev} \leq p_i \leq i'_{next} \leq i_{next}\) where \(S\rointv{i'_{prev}}{p_i} \supseteq A\) and \(S\lointv{p_i}{i'_{next}} \supseteq B\). 
    From this, it follows directly that \(\nextpt{S}{B}{p_i} \leq \nextpt{S}{B'}{p_i}\) and \(\prevpt{S}{A'}{p_i} \leq \prevpt{S}{A}{p_i}\). 
    Therefore, for all \(i \in \intv{1}{k-1}\) we have,
    \[\nextpt{S}{B}{p_i} \leq \nextpt{S}{B'}{p_i} < \prevpt{S}{A'}{p_{i+1}} \leq \prevpt{S}{A}{p_{i+1}}.\]
    Since the indices \(i'_{prev}\) and \(i'_{next}\) exist for all \(i \in \intv{1}{k}\) and \(\nextpt{S}{B}{p_i} < \prevpt{S}{A}{p_{i+1}}\) when \(i \neq k\), it follows from the \cref{def:sigmasplitpoint} that \(p\) is a \sigmasplitpoint~for \(A\) and \(B\).
\qed
\end{proof}


\begin{algorithm}
\caption{Algorithms for Extending a repeating subsequence}
\begin{algorithmic}[1]
\Require \(X\) a \(k\)-repeating subsequence of \(S\) where \(X\) is not the empty string with character \(\sigma\) appearing in \(X\).
\Function{Extend\(k\)RepSubseq}{$S, X, \sigma$}
        \State \(A \gets X[\ldots i), B \gets X\lointv{i}{}\) where \(i\) is the index of the first \(\sigma\) in X. \label{line:1}
        \State \(P(\sigma) \gets\) the indices of \(S\) where \(\sigma\) occurs.
    \For{$\alpha \in$\textsc{EnumPotential-$\sigma$-Start}$(\numberofocc{X}{\sigma},\,k,\,P(\sigma)\,)$}\label{line:main-loop}
        \If{\(\alpha\) is a \(\sigma\)-split point for \(A, B\)}
            \State \(e \gets (\prevpt{S}{A}{\alpha_2},\;\prevpt{S}{A}{\alpha_3}, \ldots, \prevpt{S}{A}{\alpha_{k-1}},\; |S|+1)\) \label{line:main-block-start}
            \State \(B \gets\)\algoextendcs\((S(\alpha_1\ldots e_1),\ldots,S(\alpha_k\ldots e_k)\;; B)\) \label{line:ExtendB}
            \State \(s \gets (0,\;\nextpt{S}{B}{\alpha_1}, \ldots, \nextpt{S}{B}{\alpha_{k}})\)
            \State \(A \gets\)\algoextendcs\((S(s_1\ldots \alpha_1),\ldots,S(s_k\ldots \alpha_k)\;; A)\) \label{line:EntendA}
        \EndIf
    \EndFor
    \State \Return \(A\sigma B\)
\EndFunction
\end{algorithmic}
    \label{algo:extend-k-repeating}
\end{algorithm}

We now provide~\cref{algo:extend-k-repeating} for extending a nonempty \(k\)-repeating subsequence in \(S\) to a maximal \(k\)-repeating subsequence and then show our main results.
In the psuedocode,
\algoextendcs~refers to
the algorithm by Hirota and Sakai which
we write as \(\algoextendcs(S_1,\ldots, S_k; X)\) to denote that it searches the strings \(S_1\) to \(S_k\) for a maximal common subsequence containing constraint \(X\).
The function \algoenumsigstrt\((r,k,I_{\sigma})\) generates all \(\sigma\)-starts for \(\sigma^r\) in \(S\) given the list of the locations \(I_{\sigma}\) of \(\sigma\) in \(S\) provided in ascending order. It can be found in~\cref{algo:starsandbars} in appendix~\ref{app:enumeratingsigstart}.

\begin{theorem}
    Let \(S\) and \(X\in \krep{S}\) be strings and \(\sigma\) appear in \(X\). 
    \textsc{Extend\(k\)}-
    \textsc{RepSubseq}\((S,\) \(X,\)\(\sigma)\) returns a 
    \(Y \in \maxkrep{S}\) such that 
    \(X \subseteq Y\) in \(O(\binom{R+k}{k}\cdot kn\log n)\) time, where \(R\) is the number of unused \(\sigma\) in \(S\) after placing \(k\) copies of \(X\). 
    \label{thm:kextendcorrectness}
\end{theorem}
\begin{proof}
Let  \(A_0, B_0 \in \Sigma^{*}\) such that \(X=A_0\sigma B_0\).
First note that the string \(A\) and \(B\) are supersequences of \(A_0\) and \(B_0\) after \cref{line:1} executes and \(A\sigma B\) is a \(k\)-repeating subsequence.
Further after each iteration of the loop on \cref{line:main-loop}, \(A\) and \(B\)
must still be supersequences as \algoextendcs
~can only produce supersequences of its constraints.
Now note that \cref{line:ExtendB} and \cref{line:EntendA} always call \algoextendcs~on
non overlapping substrings that contain a \(\sigma\) between them, thus \((A\sigma B)^k\) is contained
in the concatenation of these substrings hence is a subsequence of \(S\). Thus, \(A \sigma B\) is a \(k\)-repeating subsequence that contains \(X\) after every iteration of~\cref{line:main-loop} therefore,
the algorithm outputs a \(k\)-repeating subsequence that
contains \(X\).

Now, we will show the \(A\sigma B\) returned by the \algoextendkrep~is maximal.
Let \(Y \in \krep{S}\) such that \(A\sigma B \subseteq Y\).
By~\cref{obs:prefix-suffix-extend}, \(Y = A'\sigma B'\) where \(A \subseteq A'\) and \(B \subseteq B'\) for some strings \(A',B'\).
Further by~\cref{lemma:bicond-krep}, we have that there must exists some \(\sigma\)-split point \(p\) between \(A'\) and \(B'\).
As \(A_0 \subseteq A \subseteq A'\) and \(B_0 \subseteq B \subseteq B'\), we have \(p\) is also \(\sigma\)-split point for \(A_0\) and \(B_0\) by \cref{lemma:split-implies-split}, and \(p\) is
a \(\sigma\)-start for~\(\sigma^{\numberofocc{X}{\sigma}}\).
As \algoenumsigstrt(\(\numberofocc{X}{\sigma},\,k,\,P(\sigma)\,)\)~generates all \(\sigma\)-starts for~\(\sigma^{\numberofocc{X}{\sigma}}\)
, consider the iteration of the loop on~\cref{line:main-loop} where \(\alpha = p\).
Let \(A_p\) and \(B_p\) be values of \(A\) and \(B\) in the algorithm
before this iteration is completed and let \(A_p'\) and \(B_p'\) be their values after.
As \(B_p \subseteq B\) and \(A_p \subseteq A\), we have \(p\)
is a \(\sigma\)-split for \(A_p\) and \(B_p\) by \cref{lemma:split-implies-split} thus, \cref{line:main-block-start} to \cref{line:EntendA}
are executed this iteration.
As \(A_p \subseteq A\), we have
\(\prevpt{S}{A'}{p_i} \leq \prevpt{S}{A_p}{p_i}\) for 
every \(i \in \intv{1}{k}\).
As \(e = (\prevpt{S}{A_p}{p_2}, \ldots,\;\prevpt{S}{A_p}{p_{k-1}},\;|S|+1)\) after~\cref{line:ExtendB},
we have \(B' \subseteq S\ointv{p_i}{\prevpt{S}{A'}{p_i}} \subseteq S\ointv{p_i}{e_i}\) for \(i \in \intv{1}{k-1}
\) and \(B'\subseteq S\ointv{p_k}{|S|+1} = S\ointv{\alpha_k}{e_k}\).
Therefore \(B'\) is a common subsequence to the input of
\algoextendcs~on~\cref{line:ExtendB} which contains the contains the constraint \(B_p\).
After~\cref{line:ExtendB} is executed, the variable \(B\) is assigned to \(B_p'\) which must be
a maximal common subsequence \(S\ointv{\alpha_i}{e_i}\) for \(i \in \intv{1}{k}\).
As \(B_p' \subseteq B \subseteq B'\), by maximality, \(B_p' = B'\) containing \(B = B'\).
Showing \(A = A_p' = A'\) follows in a very similar fashion.
Therefore \(A\sigma B = A'\sigma B'= Y\).
Thus, we have the only supersequence of \(A\sigma B\) which is also a \(k\)-repeating subsequence of \(S\) is \(A\sigma B\), so \(A\sigma B\) is maximal.

Note that~\algoextendkrep~spends no more than \(O(n)\) time before it begins iterating the loop on~\cref{line:main-loop}.
Further, the time spent per iteration of the loop is dominated by the time it takes to extend \(A\) and \(B\), thus taking no more than  \(O(kn\log n)\).
Since there are no more than \(\binom{\numberofocc{S}{\sigma} - k \cdot r+k}{k}\) iterations of the loop by~\cref{lemma:EnumeratePotentialcorrectness}.
Consequently, the algorithm
runs in \(O(\binom{R+k}{k}\cdot kn\log n)\) time, where \(R = \numberofocc{S}{\sigma} - k r\). 
\qed
\end{proof}

\begin{corollary}
Given a sequence \(S\), one maximal \(k\)-repeating subsequence can be found in \(O(k\binom{2k-1}{k}\cdot n\log n)\) time.
    \label{cor:correctnessoffindone}
\end{corollary}
\begin{proof}
If \(S\) does not contain any letter repeated \(k\) times, we can detect this and return \( \epsilon \) in \( O(n + |\Sigma|) = O(n) \) time.
Otherwise, let \(\sigma\) be any character that appears \( \ell \geq k \) times in \( S \). 
Then we select the subsequence \(X=\sigma^{\lfloor \frac{\ell}{k}\rfloor}\). Observe that \(X \in \krep{S}\) holds trivially. 
Then we call \algoextendkrep~to extend \(X\) and return a maximal \(k\)-repeating subsequence in \(O(k\cdot\binom{\ell-k\cdot \lfloor \ell/k\rfloor + k}{k}\cdot n\log n)\) time as \(\ell-k\cdot \lfloor \ell/k\rfloor \leq k-1\). Thus, we have a running time of \(O(k\cdot \binom{2k-1}{k}\cdot n\log n)\).
    \label{proof:correctnessoffindone}
    \qed
\end{proof}

Note that if we define $f(k)=k\cdot \binom{2k-1}{k}$, then the above time bound becomes $O(f(k)n\log n)$. Moreover, when $k=O(1)$ a maximal \(k\)-repeating subsequence of $S$ can be computed in $O(n\log n)$ time. 



\section{Concluding Remarks}

We give efficient algorithms to compute the maximal square subsequence (resp. maximal $k$-repeating subsequence) of an input string $S$ of length $n$, with running time $O(n\log n)$ (resp. $O(f(k)n\log n)$). These are much faster than the `maximization' versions, with running time $O(n^2)$ and $O(n^{2k-1})$ respectively. 

However, for the constrained versions where a pattern $X$ of $S$ must not appear in a maximal square subsequence (resp. maximal $k$-repeating subsequence), the problems are completely open.

\section*{Acknowledgments}

This research is partially supported by NSF under project CNS-2243010 and CCF-2529539.

\newpage

\appendix

\section{Enumerating \(\sigma\)-Start}
\label{app:enumeratingsigstart}
We consider in detail an algorithm to produce \(\sigma\)-starts in this section.
We use \(\mathbf{yeild}\) in the pseudocode to denote that the algorithm provides an output but does not necessarily terminate; instead, it pauses and waits for a request for the next item to continue running.
The algorithm terminates when there are no more items to provide.
First, note \textsc{EnumPossibleDivisions} yields
every element of \(\Delta_{d;h} = \{ \Bar{x} \in \mathbb{Z}^{d+1}_{\geq 0} \colon \sum^{d+1}_{i=1}x_i = h\}\) exactly
once spending no more \(O(d)\) time per item yielded.
This function is intuitively to equivalent generating all the ways to partition \(h\) indistinguishable elements into \(d+1\) distinguishable bins.

\begin{algorithm}
\caption{Algorithms for Enumerating the division of \(k\) objects between \(d\) dividers}
\label{algo:starsandbars}
\begin{algorithmic}[1]
\Require \(k\) and \(d\) are positive integers.
\Function{EnumPossibleDivisions}{$k, d$}
  \If{$d\le0$}
    \State \textbf{yield} $k$
  \EndIf
  \For{$i \gets 0$ \textbf{to} $k$}
    \For{$sol \in$\textsc{EnumPossibleDivisions}$(k-i, d-1)$}
        \State \textbf{yield} $[\,i\,] + sol$
    \EndFor
  \EndFor
\EndFunction
\Require \(k\) is the number of repetitions \(X\) in \(S\) and \(r\) is the number of \(\sigma\).
\(I_{\sigma}\) is the list of indices of every \(\sigma\) in \(S\) in ascending order. 
 \Function{EnumPotential-$\sigma$-Starts}{$r, k, I_\sigma$}
    \For{$sol \in$\textsc{ EnumPossibleDivisions}(\(|I_\sigma|-k\cdot r\), \(k\))}
        \State \textbf{yield} \((I_{\sigma}[(t-1)r +1 +\sum_{i=1}^{t}sol[i]] \text{ for }t\text{ from } 1\text{ to }k-1)\) \label{line:yieldline9}
    \EndFor
\EndFunction
\end{algorithmic}
\end{algorithm}

\begin{lemma}
    The algorithm \textsc{EnumeratePotential-\(\sigma\)-Starts}
    yields every \(\sigma\)-start for \(\sigma^r \subseteq S\) where \(r \in \Z_{>0}\)
    exactly once and spends a total of
    \(O \left( k\cdot \binom{\numberofocc{S}{\sigma} - k \cdot r+k}{k}\right)\) time to produce every value.\label{lemma:EnumeratePotentialcorrectness2}
\end{lemma}

\begin{proof}
Let \( R = \numberofocc{S}{\sigma} - k \cdot r \) which is equal to the number of unused \(\sigma\) after we place \(k\) copies of \(\sigma^r\) in \(S\).  
Consider \( \boldsymbol{\delta} \in \intcompositons{k}{R} \) and define \( w \in \mathbb{Z}_{\geq 0}^{k} \) by:
\[
w_1 = \delta_1,
\quad \text{and} \quad
w_{i+1} = w_i + \delta_{i+1} + r 
\quad \text{for } i \in \lointv{1}{k}.
\]
Observe that:
\[
w_k = \sum_{i=1}^{k} \delta_i + r(k-1) + 1.
\]
Hence:
\[
\numberofocc{X}{\sigma} - w_k \geq r - 1.
\]
This shows that at the \( w_k \)-th occurrence of \(\sigma\) in \( S \), there are at least \( r - 1 \) additional occurrences remaining.

For each \( i \in \rointv{1}{k} \), we can select the \( w_i \)-th occurrence of \(\sigma\) in \( S \) and build an alignment of \(\sigma^r\).  
Since \( w_{i+1} - w_i \geq r \), the block of \(\sigma^r\) starting at the \( w_i \)-th occurrence will end strictly before the \( w_{i+1} \)-th occurrence.

Moreover, note that:
\[
w_i = \sum_{j=1}^{i} \delta_j + r(i-1) + 1,
\]
which is exactly the expression appearing at \cref{line:yieldline9} in \cref{algo:starsandbars}.  
Therefore, the procedure \textsc{EnumPotential-\(\sigma\)-Starts} always produces a valid \(\sigma\)-start for \(\sigma^r\).

Additionally, for every alignment \(\psi\) of \(\sigma^{rk}\), we can construct the \((k+1)\)-tuple \(q \in \Z^{k+1}_{\geq 0}\), where \(q_1 = |\{i \in H_{\sigma}:i < \psi(1)\}|\), \(q_{k+1} = |\{i \in H_{\sigma} : \psi(1+r(k-1)) < i\}|)\), and \(q_i = |\{{j \in H_{\sigma} : \psi(1 + r(i-2))} < j < \psi(1 + r(i-1))|\}\)
where \( H_{\sigma} \) is the set of indices of \(\sigma\) in \( S \) that are unused by alignment \(\psi\).
It is clear that \( q \in \intcompositons{k}{R} \), and moreover, the location of the first \(\sigma\) in each \(\sigma^r\) block in \(\psi\) can be uniquely recovered from \( q \). 
Consequently, each \(\sigma\)-start is generated exactly once.

Note that \textsc{EnumPotential-\(\sigma\)-Starts} 
generates exactly one unique output per element of \(\intcompositons{k}{R}\), with a per-element processing time of \(O(k)\).
As \(|\intcompositons{k}{R}|= \binom{R+k}{k}\),
\textsc{EnumPotential-\(\sigma\)-Starts} spends \(O\left(k\cdot \binom{R + k}{k}\right)\) time to generate every starting point.
\qed
\end{proof}

As a consequence of~\cref{lemma:split-implies-split} and the prior lemma, we have the following:

\begin{corollary}
    The algorithm \textsc{EnumeratePotential-\(\sigma\)-Starts}
    yields every \(\sigma\)-start for \(\sigma^{\numberofocc{X}{\sigma}}\) for the \(k\)-repeating subsequence \(X \subseteq S\).
    \label{lemma:EnumeratePotentialcorrectnessX}
\end{corollary}

\end{document}